\documentclass[10pt,conference]{IEEEtran}

\usepackage{amssymb}
\usepackage{amsfonts}
\usepackage[colorlinks,citecolor=blue,urlcolor=blue,breaklinks=true]{hyperref}
\usepackage{tikz}

\begin{document}

\title{Authentication and Secrecy Codes for Equiprobable Source Probability Distributions}

\author{
\authorblockN{Michael Huber}
\authorblockA{Wilhelm-Schickard-Institute for Computer Science\\
University of Tuebingen \\
Sand~13, 72076 Tuebingen, Germany\\
Email: michael.huber@uni-tuebingen.de}
}

\newtheorem{corollary}{Corollary}
\newtheorem{theorem}{Theorem}
\newtheorem{lemma}{Lemma}
\newtheorem{problem}{Problem}
\newtheorem{example}{Example}
\newtheorem{remark}{Remark}

\maketitle

\begin{abstract}
We give new combinatorial constructions for codes providing authentication and secrecy for equiprobable source probability distributions.
In particular, we construct an \mbox{infinite} class of optimal authentication codes which are multiple-fold \mbox{secure} against spoofing and simultaneously achieve perfect \mbox{secrecy}. Several further new optimal codes satisfying these properties will also be constructed and presented in general tables. Almost all of these appear to be the first authentication codes with these properties.
\end{abstract}

\section{Introduction}

The construction of authentication codes is an important topic in cryptography, and has been considered by
many researchers over the last few decades. The first construction of such codes go back to Gilbert, MacWilliams and Sloane~\cite{gil74}, using finite projective planes.

In this paper, we consider combinatorial constructions for codes providing authentication and secrecy for equiprobable source probability distributions. For general authentication codes without any secrecy requirements, there exist various constructions for a long time, regardless of the source distribution. However, if we wish that the authentication codes simultaneously provide for secrecy, then there are only a few constructions known, see e.g.~\cite{ding04,Mass86,pei06,Stin90}. These constructions are mostly of combinatorial nature, using combinatorial $t$-designs, perpendicular arrays, or orthogonal arrays. An algebraic approach~\cite{ding04} is based on (non-)linear functions between finite Abelian groups.
In particular, Stinson~\cite{Stin90} constructed in 1990 optimal authentication codes that are one-fold secure against spoofing and achieve perfect secrecy. His constructions rely on Steiner \mbox{$2$-designs} and assume that the source states are equiprobable distributed (cf.~Theorems~\ref{stin1} and~\ref{stin2}).

We will extend Stinson's constructions to obtain optimal codes which are multi-fold secure against spoofing and provide perfect secrecy. This can be achieved by means of Steiner \mbox{$t$-designs} for larger $t$. Using M\"{o}bius planes, and more generally spherical geometries, we will particularly construct a new infinite class of optimal codes which are two-fold secure against spoofing and achieve perfect secrecy (Section~\ref{new}).
Several further new optimal codes satisfying these properties will be constructed and presented in general tables (Section~\ref{table}).
Almost all of these appear to be the first authentication codes with these properties.

The paper uses concepts from both combinatorial design theory and the theory of authentication codes. Relevant definitions will be summarized
(Sections~\ref{auth} and~\ref{des}) as well as results on general authentication codes that are important for our purposes (Section~\ref{genauth}). The paper concludes with a discussion on further research directions (Section~\ref{disc}).

\section{Authentication and Secrecy Model}\label{auth}

We rely on the unconditional (theoretical) secrecy model developed by Shannon~\cite{Shan49}, and by Simmons (e.g.~\cite{Sim85}) including authentication. We follow the description and notion of~\cite{Mass86,Stin90}. We also mention the recent reference work~\cite{pei06}.

In this model of authentication and secrecy three participants are involved: a \emph{transmitter}, a \emph{receiver}, and an \emph{opponent}.  The transmitter wants to communicate information to the receiver via a public communications channel. The receiver in return would like to be confident that any received information actually came from the transmitter and not from some opponent (\emph{integrity} of information). The transmitter and the receiver are assumed to trust each other. Sometimes this is also called an \emph{$A$-code}.

Let $\mathcal{S}$ denote a set of $k$ \emph{source states} (or \emph{plaintexts}), $\mathcal{M}$ a set of $v$ \emph{messages} (or \emph{ciphertexts}), and $\mathcal{E}$ a set of $b$ \emph{encoding rules} (or \emph{keys}). Using an encoding rule $e\in \mathcal{E}$,
the transmitter encrypts a source state $s \in \mathcal{S}$ to obtain the message $m=e(s)$ to be sent over the channel. The encoding rule is an injective function from $\mathcal{S}$ to $\mathcal{M}$, and is communicated to the receiver via a secure channel prior to any messages being sent. For each encoding rule $e \in \mathcal{E}$, let $M(e):=\{e(s) : s \in \mathcal{S}\}$ denote the set of \emph{valid} messages. A received message $m$ will be accepted by the receiver as being authentic if and only if $m \in M(e)$. When this is fulfilled, the receiver decrypts the message $m$ by applying the decoding rule $e^{-1}$, where \[e^{-1}(m)=s \Leftrightarrow e(s)=m.\]
An authentication code can be represented algebraically by a $(b \times k)$\emph{-encoding matrix} with the rows indexed by the encoding rules, the columns indexed by the source states, and the entries defined by $a_{es}:=e(s)$ ($1\leq e \leq b$, $1\leq s \leq k$).

\subsection{Spoofing Attack and Perfect Secrecy}

We are interested in the following scenario, which is called
\emph{spoofing attack} of order $i$: 
Suppose that an opponent observes $i\geq 0$ distinct messages, which are sent through the public channel using the same encoding rule. The opponent then inserts a new message $m'$ (being distinct from the $i$ messages already sent), hoping to have it accepted by the receiver as authentic.
The cases $i=0$ and $i=1$ are called \emph{impersonation game} and \emph{substitution game}, respectively. These cases have been studied in detail in recent years, whereas less is known for the cases $i \geq 2$. 

We assume that there are probability distributions $p_S$ on $\mathcal{S}$ and $p_E$ on $\mathcal{E}$ (known to all participants) with associated independent random variables $S$ and $E$, respectively. These distributions induce a third distribution, $p_M$, on $\mathcal{M}$ with associated random variable $M$. The \emph{deception probability} $P_{d_i}$ is the probability that the opponent can deceive the receiver with a spoofing attack of order $i$. 

\begin{theorem}$[$Massey$]$
In an authentication code with $k$ source states and $v$ messages, the deception probabilities are bounded below by
\[P_{d_i}\geq \frac{k-i}{v-i}.\]
\end{theorem} 

An authentication code is called $t$\emph{-fold secure against spoofing} if
$P_{d_i}= (k-i)/(v-i)$ for $0 \leq i \leq t$.

In what follows, we are also interested in the property of secrecy:
An authentication code is said to have \emph{perfect secrecy} if
\[p_S(s | m)=p_S(s)\]
for every source state $s \in \mathcal{S}$ and every message $m \in \mathcal{M}$.
That is, the \emph{a posteriori} probability that the source state is $s$, given that the message $m$ is observed, is identical to
the \emph{a priori} probability that the source state is $s$. 
It can easily be shown that
\[p_S(s|m) =\frac{\sum_{\{e \in \mathcal{E}: e(s)=m\}} p_E(e)p_S(s)}{\sum_{\{e \in \mathcal{E}:m\in M(e)\}} p_E(e)p_S(e^{-1}(m))}.\]
As a consequence, we have
\begin{lemma}$[$Stinson$]$\label{frequency}
An authentication code has perfect secrecy if and only if
\[\sum_{\{e \in \mathcal{E}: e(s)=m\}} p_E(e)=\sum_{\{e \in \mathcal{E}:m\in M(e)\}} p_E(e)p_S(e^{-1}(m))\]
for every source state $s$ and every message $m$.
\end{lemma}

Thus, if the encoding rules in a code are used with equal probability, then a given message $m$ occurs with the same frequency in each column of the encoding matrix.

\section{Combinatorial Design Theory}\label{des}

For positive integers $t \leq k \leq v$ and $\lambda$, a
\mbox{\emph{$t$-$(v,k,\lambda)$ design}} $\mathcal{D}$ is a pair \mbox{$(X,\mathcal{B})$}, satisfying
the following properties:

\begin{enumerate}
\item[(i)] $X$ is a set of $v$ elements, called \emph{points},

\item[(ii)] $\mathcal{B}$ is a family of \mbox{$k$-subsets} of $X$, called \emph{blocks},

\item[(iii)] every \mbox{$t$-subset} of $X$ is contained in exactly $\lambda$ blocks.

\end{enumerate}
We will denote points by lower-case and blocks by upper-case Latin
letters.
Via convention, let $b:=\left| \mathcal{B} \right|$ denote the number of blocks.
Throughout this article, `repeated
blocks' are not allowed, that is, the same \mbox{$k$-subset}
of points may not occur twice as a block. If $t<k<v$ holds, then we
speak of a \emph{non-trivial} \mbox{$t$-design}.
For historical reasons, a \mbox{$t$-$(v,k,\lambda)$ design} with
$\lambda =1$ is called a \emph{Steiner \mbox{$t$-design}} (sometimes
also a \emph{Steiner system}).
As a simple example, let us choose as point set $X=\{1,2,3,4,5,6,7\}$ and as block set
$\mathcal{B}=\{\{1,2,4\},\{2,3,5\},\linebreak\{3,4,6\},\{4,5,7\},\{1,5,6\},\{2,6,7\},\{1,3,7\}\}.$
This gives a Steiner \mbox{$2$-$(7,3,1)$ design}, the well-known \emph{Fano
plane}, which is the smallest design arising from a
projective geometry. The usual representation of this unique projective plane of order $2$ is given by the
following diagram: 

\begin{figure}[htp]\label{fano}

\centering

\begin{tikzpicture}[scale=1,thick]

\filldraw [draw=black!100,fill=black!0]

(0,0.87) circle (24.4pt);

\filldraw [draw=black!100,fill=black!100]

(0,0) circle (2pt)  (0,-0.3) node {7}

(-1.5,0) circle (2pt) (-1.52,-0.3) node {1}

(1.5,0) circle (2pt) (1.52,-0.3) node {3}

(0,0.87) circle (2pt) (0.12,1.15) node {6}

(0,2.6) circle(2pt) (0,2.9) node {2}

(0.75,1.31) circle (2pt) (0.95,1.35) node {5}

(-0.75,1.31) circle(2pt) (-0.95,1.35) node {4};

\draw (1.5,0) -- (-1.5,0)

(1.5,0) -- (0,2.6)

(-1.5,0) -- (0,2.6)

(-1.5,0) -- (0.75,1.31)

(1.5,0) -- (-0.75,1.31)

(0,0) -- (0,2.6);

\end{tikzpicture}
\caption{The Fano plane}\label{fano}
\end{figure}
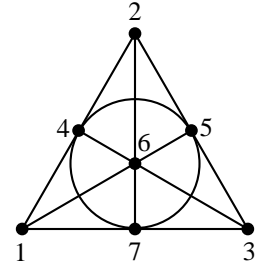

If a \mbox{$2$-design} has, like in this example, equally many points and blocks, i.e.
$v=b$, then we usually speak of a \emph{symmetric design}.
For the existence of \mbox{$t$-designs}, basic necessary
conditions can be obtained via elementary counting arguments (see,
for instance,~\cite{BJL1999}):

\begin{lemma}\label{s-design}
Let $\mathcal{D}=(X,\mathcal{B})$ be a \mbox{$t$-$(v,k,\lambda)$}
design, and for a positive integer $s \leq t$, let $S \subseteq X$
with $\left|S\right|=s$. Then the number of blocks containing
each element of $S$ is given by
\[\lambda_s = \lambda \frac{{v-s \choose t-s}}{{k-s \choose t-s}}.\]
In particular, for $t\geq 2$, a \mbox{$t$-$(v,k,\lambda)$} design is
also an \mbox{$s$-$(v,k,\lambda_s)$} design.
\end{lemma}

It is customary to set $r:= \lambda_1$ denoting the
number of blocks containing a given point (referring to the
`replication number' from statistical design of experiments, one of
the origins of design theory). It follows

\begin{lemma}\label{Comb_t=5}
Let $\mathcal{D}=(X,\mathcal{B})$ be a \mbox{$t$-$(v,k,\lambda)$}
design. Then the following holds:
\begin{enumerate}

\item[{(a)}] $bk = vr.$

\smallskip

\item[{(b)}] $\displaystyle{{v \choose t} \lambda = b {k \choose t}.}$

\smallskip

\item[{(c)}] $r(k-1)=\lambda_2(v-1)$ for $t \geq 2$.

\end{enumerate}
\end{lemma}

There are many infinite classes of Steiner \mbox{$t$-designs} for $t=2$ and $3$, however for $t=4$ and
$5$ only a finite number are known. For encyclopedic accounts of key results in design theory as well as
existence tables with known parameter sets, we refer to~\cite{BJL1999,crc06}.

\begin{problem}
Does there exist any non-trivial Steiner \mbox{$t$-design} with $t \geq 6$?
\end{problem}

\section{General Authentication Codes}\label{genauth}

For \emph{general} authentication codes, no secrecy requirements are specified. We summarize the state-of-the-art with respect to our further purposes:

The following theorem (cf.~\cite{Mass86,Sch86}) gives a lower bound on the number of encoding rules for \emph{any} source probability distribution.

\begin{theorem}$[$Massey--Sch\"{o}bi$]$
If a general authentication code is $(t-1)$-fold against spoofing, then the number of encoding rules is bounded below by
\[b \geq \frac{{v \choose t}}{{k \choose t}}.\]
\end{theorem}

As usual, we call a code \emph{optimal} if the number of encoding rules meets the lower bound with equality. When the source states are known to be independent and equiprobable, optimal authentication codes which are $(t-1)$-fold against spoofing can be constructed via \mbox{$t$-designs} (cf.~\cite{DeS88,Sch86,Stin90}):

\begin{theorem}$[$DeSoete--Sch\"{o}bi--Stinson$]$\label{general}
Suppose there is a \mbox{$t$-$(v,k,\lambda)$} design. Then there is an authentication code for $k$ equiprobable source states, having $v$ messages and $\lambda \cdot {v \choose t}/{k \choose t}$ encoding rules, that is $(t-1)$-fold secure against spoofing. Conversely, if there is an authentication code for $k$ equiprobable source states, having  $v$ messages and ${v \choose t}/{k \choose t}$ encoding rules, that is $(t-1)$-fold secure against spoofing, then there is a Steiner \mbox{$t$-$(v,k,1)$} design.
\end{theorem}

\section{Authentication Codes with Perfect Secrecy}\label{new}

Stinson~\cite[Thm.\,6.4]{Stin90} constructed in 1990 the first optimal authentication codes that are one-fold secure against spoofing and simultaneously achieve perfect secrecy. His constructions rely on Steiner \mbox{$2$-designs} and assume that the source states are equiprobable distributed.

\begin{theorem}$[$Stinson$]$\label{stin1}
Suppose there is a Steiner \mbox{$2$-$(v,k,1)$} design, where $v$ divides the number of blocks $b$.
Then there is an optimal authentication code for $k$ equiprobable source states, having $v$ messages and $v(v-1)/k(k-1)$ encoding rules, that is one-fold secure against spoofing and provides perfect secrecy.
\end{theorem}

Using Steiner $2$-$(\frac{q^{d+1}-1}{q-1},q+1,1)$ designs whose points and blocks
are the points and lines of projective spaces $PG(d,q)$, Stinson~\cite[Thm.\,6.5]{Stin90} constructed this way an infinite class of authentication codes with the following properties:

\begin{theorem}$[$Stinson$]$\label{stin2}
For all prime powers $q$ and for all even $d \geq 2$, there is an optimal authentication code for an equiprobable source probability distribution with $q+1$ source states, having $(q^{d+1}-1)/(q-1)$ messages, that is one-fold secure against spoofing and provides perfect secrecy.
\end{theorem}

The smallest example is as follows (cf.~\cite[Ex.\,6.1]{Stin90}):

\begin{example}
An optimal authentication code for $k=3$ equiprobable source states, having $v=7$ messages, and $b=7$ encoding rules, that is one-fold secure against spoofing and provides perfect secrecy can be constructed from a Steiner \mbox{$2$-$(7,3,1)$ design}, i.e. the unique Fano plane illustrated in Fig.~\ref{fano}. Each encoding rule is used with probability $1/7$. An encoding matrix is given in Table~\ref{STS}.

\begin{table}
\renewcommand{\arraystretch}{1.3}
\caption{Authentication code from the Fano plane}\label{STS}

\begin{center}

\begin{tabular}{|c| c c c|}
  \hline
  & $s_1$ & $s_2$ & $s_3$ \\
  \hline
  $e_1$ & 1 & 2 & 4 \\
  $e_2$ & 2 & 3 & 5 \\
  $e_3$ & 3 & 4 & 6 \\
  $e_4$ & 4 & 5 & 7 \\
  $e_5$ & 5 & 6 & 1 \\
  $e_6$ & 6 & 7 & 2 \\
  $e_7$ & 7 & 1 & 3 \\
  \hline
\end{tabular}

\end{center}
\end{table}

\end{example}

We will extend Stinson's constructions to obtain optimal codes which are multi-fold secure against spoofing and provide perfect secrecy. This can be achieved by means of Steiner \mbox{$t$-designs} for larger $t$.

\begin{theorem}\label{mythm1}
Suppose there is a Steiner \mbox{$t$-$(v,k,1)$} design, where $v$ divides the number of blocks $b$.
Then there is an optimal authentication code for $k$ equiprobable source states, having $v$ messages and ${v \choose t}/{k \choose t}$
encoding rules, that is $(t-1)$-fold secure against spoofing and provides perfect secrecy.
\end{theorem}

\begin{proof}
Let $\mathcal{D}=(X,\mathcal{B})$ be a Steiner \mbox{$t$-$(v,k,1)$} design, where  $v$ divides $b$. We mimic the proof of \cite[Thm.\,6.4]{Stin90}. Clearly, the authentication capacity of the code follows via Theorem~\ref{general}. To establish perfect secrecy under the assumption that the encoding rules are used with equal probability, it is necessary in view of Lemma~\ref{frequency} that a given message occurs with the same frequency in each column of the resulting encoding matrix.
This can be done by ordering every block of $\mathcal{D}$ in such a way that every point occurs in each possible position in precisely $b/v$ blocks.
Since every point occurs in exactly $r={v-1 \choose t-1}/{k-1 \choose t-1}$ blocks due to Lemma~\ref{Comb_t=5}~(c), necessarily clearly $k$ must divide $r$. To show that the condition is also sufficient, we consider the bipartite point-block incidence graph of $\mathcal{D}$ with vertex set $X \cup \mathcal{B}$, where $(x,B)$ is an edge if and only if $x \in B$  for $x \in X$ and $B \in \mathcal{B}$. An ordering on each block of $\mathcal{D}$ can be obtained via an edge-coloring of this graph using $k$ colors in such a way that each vertex $B \in \mathcal{B}$ is adjacent to one edge of each color, and each vertex $x \in X$ is adjacent to $b/k$ edges of each color. Technically, this can be achieved by first splitting up each vertex $x$ into $b/k$ copies, each having degree $k$, and then by finding an appropriate edge-coloring of the resulting $k$-regular bipartite graph using $k$ colors. Taking the  ordered blocks as encoding rules, each used with equal probability, establishes the claim.
\end{proof}

Relying on M\"{o}bius planes, and more generally on spherical geometries, we can now construct a new infinite class of optimal codes which are two-fold secure against spoofing and achieve perfect secrecy.

\begin{theorem}\label{mythm2}
For all prime powers $q$ and for all even $d \geq 2$, there is an optimal authentication code for an equiprobable source probability distribution with $q+1$ source states, having $q^d+1$ messages, that is two-fold secure against spoofing and provides perfect secrecy.
\end{theorem}

\begin{proof}
Steiner designs can be constructed from spherical geometries as follows:
Let $q$ be a prime power, and $d \geq 2$ an integer. As point set $X$ choose the elements of the
projective line \mbox{$GF(q^d) \cup \{\infty\}$} over the Galois field $GF(q^d)$, where $\infty$ denotes a symbol with $\infty \notin GF(q^d)$. The
linear fractional group
\[PGL(2,q^d)=\{x \mapsto \textstyle{\frac{ax+b}{cx+d}}: a,b,c,d \in
GF(q^d), ad-bc \not= 0 \}\] acts on \mbox{$GF(q^d) \cup \{\infty\}$}
in the natural manner (with the usual conventions for $\infty$). As
block set $\mathcal{B}$ take the images of \mbox{$GF(q) \cup \{\infty\}$} under
$PGL(2,q^d)$. This gives a \mbox{$3$-$(q^d+1,q+1,1)$} design with
$PGL(2,q^d)$ as group of automorphisms. These designs were first described by
Witt~\cite{Witt38}. For $d=2$, they are often called \emph{M\"{o}bius planes}
(or \emph{inversive planes}) of order $q$.

Now, if we assume that $d$ is even, then \[q^2-1 \mid q^d-1\] and hence \[(q+1)q(q-1) \mid q^d(q^d-1).\] Therefore, $v$ divides $b$, and the claim follows by applying Theorem~\ref{mythm1}.
\end{proof}

We present the smallest example:

\begin{example}
An optimal authentication code for $k=4$ equiprobable source states, having $v=10$ messages, and $b=30$ encoding rules, that is two-fold secure against spoofing and provides perfect secrecy can be constructed from a Steiner \mbox{$3$-$(10,4,1)$ design}, i.e. the unique M\"{o}bius plane of order $3$. Each encoding rule is used with probability $1/30$. We give an encoding matrix in Table~\ref{SQS}.

\begin{table}
\renewcommand{\arraystretch}{1.3}
\caption{Authentication code from the M\"{o}bius plane of order $3$}\label{SQS}

\begin{center}

\begin{tabular}{|c| c c c c|}
  \hline
  & $s_1$ & $s_2$ & $s_3$ & $s_4$\\
  \hline
     $e_1$ & 1 & 2 & 4 & 5\\
     $e_2$ & 2 & 3 & 5 & 6\\
     $e_3$ & 3 & 4 & 6 & 7\\
     $e_4$ & 4 & 5 & 7 & 8\\
     $e_5$ & 5 & 6 & 8 & 9\\
     $e_6$ & 6 & 7 & 9 & 0\\
     $e_7$ & 7 & 8 & 0 & 1\\
     $e_8$ & 8 & 9 & 1 & 2\\
     $e_9$ & 9 & 0 & 2 & 3\\
  $e_{10}$ & 0 & 1 & 3 & 4\\
  $e_{11}$ & 1 & 2 & 3 & 7\\
  $e_{12}$ & 2 & 3 & 4 & 8\\
  $e_{13}$ & 3 & 4 & 5 & 9\\
  $e_{14}$ & 4 & 5 & 6 & 0\\
  $e_{15}$ & 5 & 6 & 7 & 1\\
  $e_{16}$ & 6 & 7 & 8 & 2\\
  $e_{17}$ & 7 & 8 & 9 & 3\\
  $e_{18}$ & 8 & 9 & 0 & 4\\
  $e_{19}$ & 9 & 0 & 1 & 5\\
  $e_{20}$ & 0 & 1 & 2 & 6\\
  $e_{21}$ & 1 & 3 & 5 & 8\\
  $e_{22}$ & 2 & 4 & 6 & 9\\
  $e_{23}$ & 3 & 5 & 7 & 0\\
  $e_{24}$ & 4 & 6 & 8 & 1\\
  $e_{25}$ & 5 & 7 & 9 & 2\\
  $e_{26}$ & 6 & 8 & 0 & 3\\
  $e_{27}$ & 7 & 9 & 1 & 4\\
  $e_{28}$ & 8 & 0 & 2 & 5\\
  $e_{29}$ & 9 & 1 & 3 & 6\\
  $e_{30}$ & 0 & 2 & 4 & 7\\
  \hline
\end{tabular}

\end{center}
\end{table}

\end{example}

\begin{remark}
We mention that the group $PGL(2,q^d)$ acts transitively on incident point-block
pairs, i.e. on the \emph{flags}, of the Steiner \mbox{$3$-$(q^d+1,q+1,1)$} design.
This reveals a high degree of regularity of the combinatorial structure. Basically all flag-transitive Steiner $t$-designs have been
determined recently, see~\cite{Hu2008}. Various further interactions between highly regular combinatorial structures and applications in information and coding theory can be found, e.g., in~\cite{Hu2009}.
\end{remark}

\section{Further Constructions}\label{table}

We will construct several further new optimal authentication codes for equiprobable source distributions, which are $(t-1)$-fold secure against spoofing and simultaneously achieve perfect secrecy. All codes with $t \geq 3$ appear to be the first authentication codes satisfying these properties.

In view of Theorem~\ref{mythm1}, we have to check whether the parameters of known Steiner \mbox{$t$-$(v,k,1)$} designs satisfy the condition that $v$ divides the number of blocks $b={v \choose t}/{k \choose t}$.\linebreak
We recall that there are two infinite classes of optimal authentication codes: one arises from projective geometries (Theorem~\ref{stin2}) and the other from spherical geometries (Theorem~\ref{mythm2}). We can construct further infinite families of optimal codes for a \emph{fixed} number of source states as follows:
\begin{itemize}
\item A Steiner \mbox{$2$-$(v,3,1)$} design (so-called \emph{Steiner triple system}) exists if and only if $v \equiv 1$ or $3$ (mod $6$). Hence, if $v \equiv 1$ (mod $6$), then an optimal authentication code can be constructed for $k=3$ equiprobable source states, having $v$ messages, and $v(v-1)/6$ encoding rules, that is one-fold secure against spoofing and provides perfect secrecy.

\item A Steiner \mbox{$2$-$(v,4,1)$} design exists if and only if $v \equiv 1$ or $4$ (mod $12$). Hence, if $v \equiv 1$ (mod $12$), then an optimal authentication code can be constructed for $k=4$ equiprobable source states, having $v$ messages, and $v(v-1)/12$ encoding rules, that is one-fold secure against spoofing and provides perfect secrecy.

\item A Steiner \mbox{$2$-$(v,5,1)$} design exists if and only if $v \equiv 1$ or $5$ (mod $20$). Hence, if $v \equiv 1$ (mod $20$), then an optimal authentication code can be constructed for $k=5$ equiprobable source states, having $v$ messages, and $v(v-1)/20$ encoding rules, that is one-fold secure against spoofing and provides perfect secrecy.

\item A Steiner \mbox{$3$-$(v,4,1)$} design (so-called \emph{Steiner quadruple system}) exists if and only if $v \equiv 2$ or $4$ (mod $6$). Hence, if $v \equiv 2$ (mod $24$), then an optimal authentication code can be constructed for $k=4$ equiprobable source states, having $v$ messages, and $v(v-1)(v-2)/24$ encoding rules, that is two-fold secure against spoofing and provides perfect secrecy.
\end{itemize}

We present further optimal codes that are $(t-1)$-fold secure against spoofing and achieve perfect secrecy in Table~\ref{t-des}.
We give the parameters of the authentication codes as well as of the respective Steiner \mbox{$t$-designs}.
All presently known Steiner \mbox{$4$-designs} and \mbox{$5$-designs} have been examined; for Steiner \mbox{$2$-designs} and \mbox{$3$-designs} only the cases up to $v=30$ have been investigated. We refer to~\cite{BJL1999,crc06} for further information on the respective designs.

\begin{table}
\renewcommand{\arraystretch}{1.3}
\caption{Further authentication codes from Steiner $t$-designs}\label{t-des}

\begin{center}
\begin{tabular}{|c||c c c| c|c|}
  \hline
  $t$ & $k$ & $v$ & $b$ & \mbox{Design Parameters} & \mbox{Design Reference}\\
  \hline \hline
3 & 5  & 26 & 260 & $3$-$(26,5,1)$ & \mbox{Denniston design}  \\
   \hline
   & 5  & 11 & 66 & $4$-$(11,5,1)$ & \mbox{Witt design} \\
   & 7  & 23 & 253 & $4$-$(23,7,1)$ & \mbox{Witt design} \\
   & 5  & 23 & 1.771 & $4$-$(23,5,1)$ & \mbox{Denniston design}\\
   & 5  & 47 & 35.673 & $4$-$(47,5,1)$ & \mbox{Denniston design} \\
 4 & 5  & 83 & 367.524 & $4$-$(83,5,1)$ & \mbox{Denniston design} \\
   & 5  & 71 & 194.327 & $4$-$(71,5,1)$ & \mbox{Mills design} \\
   & 5  & 107 & 1.032.122 & $4$-$(107,5,1)$ & \cite{crc06} \\
   & 5  & 131 & 2.343.328 & $4$-$(131,5,1)$ & \cite{crc06} \\
   & 5  & 167 & 6.251.311 & $4$-$(167,5,1)$ & \cite{crc06} \\
   & 5  & 243 & 28.344.492 & $4$-$(243,5,1)$ & \cite{crc06} \\
   \hline
   & 6  & 12 & 132 & $5$-$(12,6,1)$ & \mbox{Witt design}  \\
5  & 6  & 84 & 5.145.336 & $5$-$(84,6,1)$ & \mbox{Denniston design} \\
   & 6  & 244 & 1.152.676.008 & $5$-$(244,6,1)$ &  \cite{crc06} \\
  \hline
\end{tabular}
\end{center}
\end{table}

\section{Discussion}\label{disc}

It would be interesting for further research to examine authentication codes that are $(t-1)$-fold secure against spoofing and achieve two-fold perfect secrecy, and more generally $(t-1)$-fold perfect secrecy. For this the following condition must be satisfied: for every $t^* \leq t-1$, for every set $M^*$ of $t^*$ messages observed in the channel, and for every set $S^*$ of $t^*$ source states, we have $p(S^*|M^*)=p(S^*)$.

\section*{Acknowledgment}

The author would like to thank Doug Stinson for helpful discussions.
The author gratefully acknowledges support of his work by the Deutsche
Forschungsgemeinschaft (DFG) via a Heisenberg grant (Hu954/4) and a Heinz Maier-Leibnitz Prize grant (Hu954/5).

\end{document}